\documentclass[11pt,onecolumn,twoside]{IEEEtran}

\usepackage{amsmath,amssymb}    
\usepackage[dvips]{graphicx}  
\usepackage{verbatim}   
\usepackage{color}      
\usepackage{subfigure}  
\usepackage{hyperref}   
\usepackage{epsfig}
\usepackage{float}
\usepackage{cite}

\newcommand{\bi}{\begin{itemize}}
\newcommand{\ei}{\end{itemize}}
\newcommand{\beq}{\begin{equation}}
\newcommand{\eeq}{\end{equation}}
\newcommand{\bqn}{\begin{eqnarray*}}
\newcommand{\eqn}{\end{eqnarray*}}
\newcommand{\ba}{\begin{array}}
\newcommand{\ea}{\end{array}}
\newcommand{\bs}{\begin{small}}
\newcommand{\es}{\end{small}}
\newcommand{\nn}{\nonumber}

\newtheorem{theorem}{Theorem}[section]

\newtheorem{proposition}[theorem]{Proposition}

\newenvironment{proof}[1][Proof]{\begin{trivlist}
\item[\hskip \labelsep {\bfseries #1}]}{\end{trivlist}}

\newcommand{\qed}{\nobreak \ifvmode \relax \else
      \ifdim\lastskip<1.5em \hskip-\lastskip
      \hskip1.5em plus0em minus0.5em \fi \nobreak
      \vrule height0.75em width0.5em depth0.25em\fi}
      
\def\defeq{{\stackrel{\Delta}{=}}}

\DeclareMathOperator*{\argmin}{arg\,min}

\begin{document}
\title{Distributed Inference in Tree Networks using Coding Theory}
\author{Bhavya~Kailkhura\thanks{The authors are with the Department of Electrical Engineering and Computer Science, Syracuse University, Syracuse, NY USA. 

(email: bkailkhu@syr.edu; avempaty@syr.edu; varshney@syr.edu)

}, Aditya~Vempaty, and Pramod~K.~Varshney}
\date{}
\maketitle

\begin{abstract}
In this paper, we consider the problem of distributed inference in tree based networks. In the framework considered in this paper, distributed nodes make a 1-bit local decision regarding a phenomenon before sending it to the fusion center (FC) via intermediate nodes. We propose the use of coding theory based techniques to solve this distributed inference problem in such structures. Data is progressively compressed as it moves towards the FC. The FC makes the global inference after receiving data from intermediate nodes. Data fusion at nodes as well as at the FC is implemented via error correcting codes. In this context, we analyze the performance for a given code matrix and also design the optimal code matrices at every level of the tree. We address the problems of distributed classification and distributed estimation separately and develop schemes to perform these tasks in tree networks. The proposed schemes are of practical significance due to their simple structure. We study the asymptotic inference performance of our schemes for two different classes of tree networks: fixed height tree networks, and fixed degree tree networks. We show that the proposed schemes are asymptotically optimal under certain conditions.

\end{abstract}
\section{Introduction}
Wireless Sensor Networks (WSNs) have attracted much interest in recent years \cite{akyildiz_commag02}. Detection, classification, or estimation of certain events, targets, or phenomena, in a region of interest, is an important application of sensor networks. Different aspects of this problem have been investigated by the research community over the last few decades~\cite{Viswanathan,veer} mostly in the context of the parallel network topology. In such a framework, due to power and bandwidth constraints, each node, instead of sending its raw data, sends quantized data to a central observer or fusion center (FC). The FC combines these local nodes' data to make a global inference~\cite{Varshney:book}.
Given a parallel topology, the objective is to find efficient quantization rules for the nodes and efficient inference rule for the FC, which maximize the global performance at the FC.
Note that, in general, the problem of designing optimal inference rules is computationally expensive (NP-hard)~\cite{JNT}. 

For example, in a distributed detection framework, under the assumption of conditional independence, the optimal decision rule for each node takes the form of a likelihood ratio test, with a suitably chosen threshold. However, finding the optimal thresholds requires the solution of a system of non-linear equations and, therefore, the problem is difficult to solve, even for the network of moderate size. The analysis of optimal detection system performance is tractable only in asymptotic regime. 
It has been shown that the use of identical thresholds is asymptotically optimal~\cite{tsit}.
Under the assumption of identical thresholds, several authors have considered the problem of designing optimal decision rules in the past~\cite{Zhang,Shi,Kailkhura}.

In contrast to the distributed detection problem, in a classification problem, each decision is usually represented by ${\log}_2 M$ information
bits, where $M$ is the number of classes to be distinguished.
The problem of classification using ${\log}_2 M$ information bits has been studied for parallel topology \cite{BaekB95}. Due to bandwidth constraints, it is desirable that the local node decisions are sent to the FC with as few bits as possible. To overcome this problem, distributed classification has been proposed in which the local nodes make 1-bit (rather than ${\log}_2 M$ bit) local decisions and send them to the FC~\cite{Zhangm,Zhu2004,Wang_jsac05}. The FC then uses the local decisions collectively and makes a global inference about the underlying phenomenon. 

In~\cite{DLI,xiaohong}, the authors consider the problem of parameter estimation in a parallel topology.
Received signal strength based methods have been proposed which employ least-squares or maximum likelihood (ML) based parameter estimation techniques. These techniques are not suitable for power and bandwidth constrained networks. To overcome these drawbacks, distributed parameter estimation using quantized measurements has been addressed in \cite{RibeiroG2006a,RibeiroG2006b,NiuV2006}. Similar to the problem of distributed detection, the system design issues of distributed estimation have also been addressed only in certain scenarios, such as in \cite{VempatyCV2013}, where it has been shown that identical quantizers are optimal under certain conditions. To simplify things, in~\cite{adierror}, coding theory based iterative schemes were proposed for target localization using parallel topology where at every iteration, the FC solves an $M$-ary hypothesis testing problem and decides the region of interest for the next iteration.

Even though the parallel topology has received significant attention, there are many practical situations where parallel topology cannot be implemented due to several factors, such as, the FC being outside the communication range of the nodes and limited energy budget of the nodes~\cite{Lin}. 
In such cases, a multi-hop network is employed, where nodes are organized hierarchically into multiple levels (tree networks).
Some examples of tree networks include wireless sensor and military communication networks.
For instance, the IEEE 802.15.4 (Zigbee) specifications \cite{Aliance} and IEEE 802.22b \cite{802_22bDraft}
can support tree based topologies.

There have been limited attempts to address the distributed inference problems in tree networks~\cite{TayTW:IT08,ZhangCPMH13,Kailkhura2013,bhavyaj}. In all but the simplest cases, optimal strategies in tree based networks are difficult to derive. Most of the work on tree networks focuses on person-by-person optimal (PBPO) strategies~\cite{TayTW:IT08,ZhangCPMH13,Kailkhura2013,bhavyaj}. Also, the above works address the problem of distributed detection in tree networks while, to the best of our knowledge, the problem of distributed estimation in tree networks has not received any attention. Due to the complexity of classification and estimation in tree networks as compared to detection, these problems have been left unexplored by researchers.

In this paper, we take a first step to address the distributed inference (classification and estimation) problems in tree networks by developing an analytically tractable framework.
We first consider the distributed classification problem and propose to use coding theory based techniques to solve the problem. We analyze the asymptotic classification performance of our
scheme for two different classes of tree networks: fixed height tree networks, and fixed degree tree networks. We show that the proposed scheme is asymptotically optimal under certain conditions. Building on these results, we extend our scheme to consider the distributed estimation problem and analyze the asymptotic estimation performance of our
scheme. The remainder of the paper is organized as follows. In Section~\ref{sec:prel}, we describe the system architecture and present a brief overview of Distributed Classification Fusion using Error Correcting Codes (DCFECC) scheme~\cite{Wang_jsac05} which serves as a foundation for the schemes presented in this paper. We propose our basic coding scheme for distributed classification in tree networks in Section~\ref{sec:approachclass}. The performance of the proposed scheme in the asymptotic regime is also analyzed. We present some numerical results to gain insights into the solution. We extend this scheme for distributed estimation in tree networks in Section~\ref{sec:approachest} by formulating the estimation problem as a sequence of $M$-ary classification problems. The performance of the proposed scheme in the asymptotic regime is analyzed and some numerical results are presented. We also provide a technique for optimal region splitting for distributed estimation. 
Finally, we conclude our paper in Section~\ref{sec:disc} with some discussion on possible future work.
\section{Preliminaries}
\label{sec:prel}

\subsection{General Network Architecture}
\label{sec:syst}
Consider a perfect tree, $T(K,\;N)$, rooted at the FC (Please see Figure~\ref{syst}). Nodes at level $k$, for $1\leq k\leq K-1$, are referred to as intermediate nodes and nodes at the last level of the tree, i.e., $k=K$, are called the leaf nodes. In a perfect tree, all the intermediate nodes have an equal number of immediate successors and the number of such successors $N$ is referred to as the degree of the tree. 

\begin{figure}[t]
  \centering
    \includegraphics[height=2in, width=!]{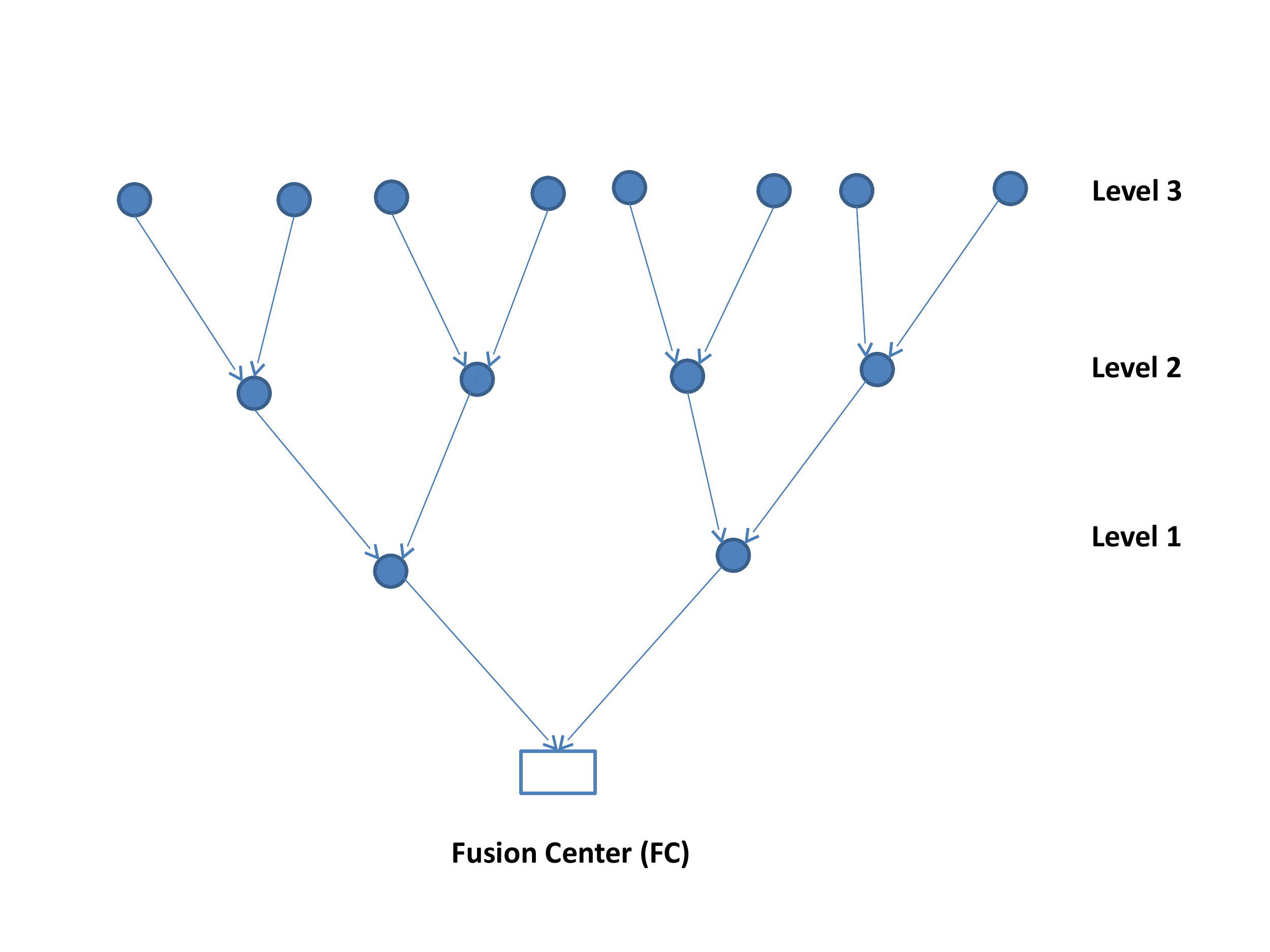}
    \vspace*{-0.3in}
    \caption{A distributed inference system organized as a perfect binary tree; $T(3,\;2)$ is shown as an example.}\label{syst}
\vspace*{-0.1in}
\end{figure}

We assume that the network is designed to infer about a particular phenomenon.
Each node $j$ at level $k$ performs two basic operations: 
\begin{itemize}
\item Depending on the task, sense data regarding the phenomenon and/or collect data from its successors at level $k+1$, denoted by $S^{k+1}(j)$.
\item Compress the data available at node $j$ about the phenomenon and transmit a 1-bit version to its predecessor at level $k-1$, denoted by $P^{k-1}(j)$. 
\end{itemize}

Local observation of node $j$ at level $k$ is denoted as $y_j^k$. Received data vector at node $j$ of level $k$ from its successors $S^{k+1}(j)$ at level $k+1$ is denoted as $\mathbf{v_j^k}\in\{0,1\}^N$. After processing the data at the node according to a processing model (Please see Figure~\ref{syst1}),
every node $j$ at level $k$ sends its one-bit local decision $u_j^k\in\{0,1\}$ to its immediate predecessor. This processing model is designed based on the inference problem considered, i.e., Figure~\ref{syst2} for classification or Figure~\ref{syst3} for estimation. Finally, the FC receives the inference vector $\mathbf{u^1}=(u_{1}^{1},\cdots,u_{N}^{1})\in\{0,1\}^N$ and fuses this data to infer about the underlying phenomenon. In our analysis, we consider error-free links in the network. However, we do provide some simulation results for the case where there are erroneous links, to examine the robustness of the proposed schemes.

Given a tree network, our objective is to find the appropriate processing scheme for nodes at all levels depending on the inference problem considered. Next, we describe the distributed classification fusion using error-correcting codes scheme (originally proposed for parallel topology in \cite{Wang_jsac05}) which serves as the mathematical basis for the ideas proposed in this paper.

\begin{figure}[t]
  \centering
    \includegraphics[height=2in, width=!]{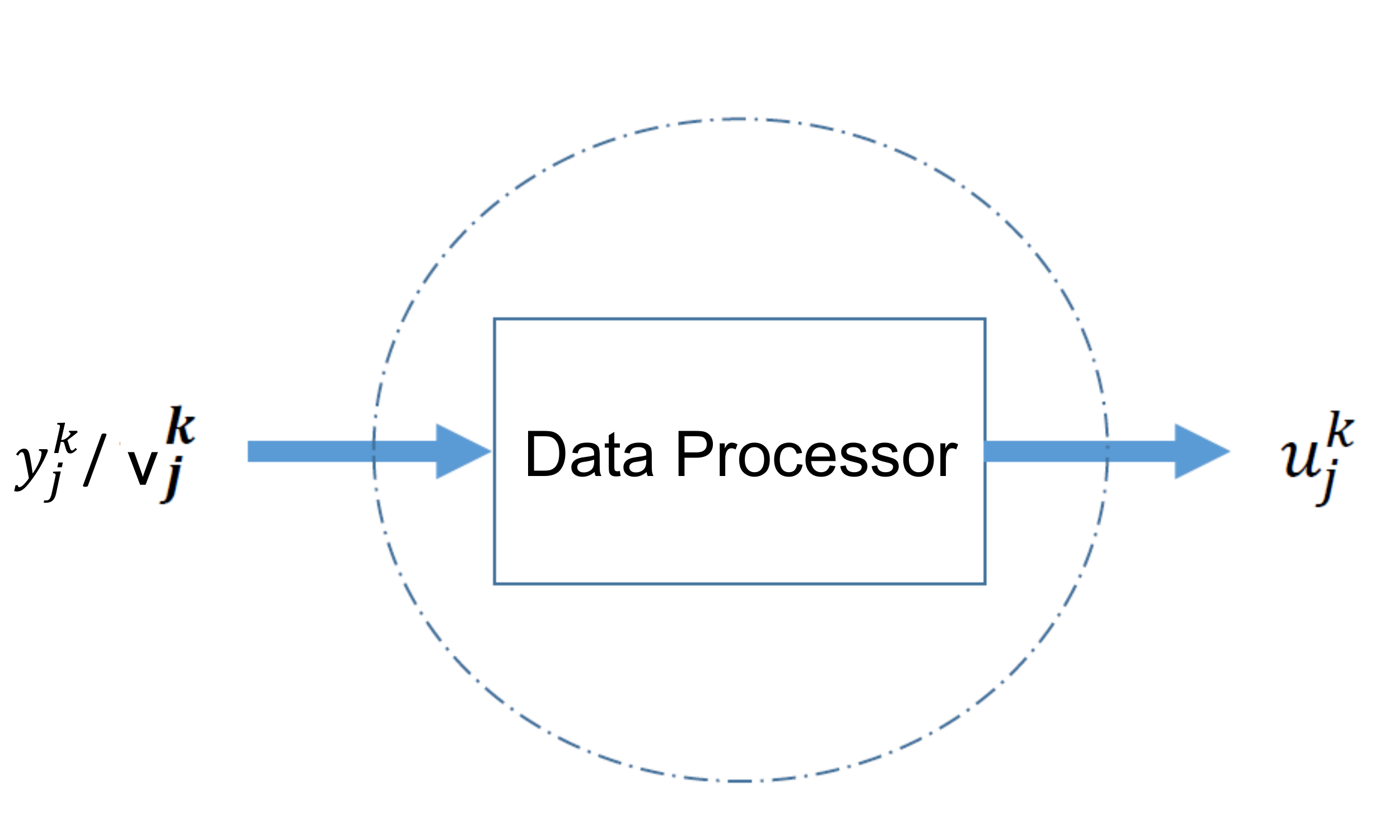}
    \caption{Data processing for distributed inference at node $j$ at level $k$. Here $y_j^k$ and $\mathbf{v_j^k}$ are the inputs and $u_j^k\in\{0,1\}$ is the output of the process at node $j$.}\label{syst1}
\end{figure}

\subsection{Distributed Classification Fusion using Error-Correcting Codes (DCFECC)}
\label{sec:DCFECC}
In~\cite{Wang_jsac05}, the authors proposed the DCFECC scheme for $M$-ary distributed classification using binary quantized local data for a parallel topology network. The idea behind the DCFECC scheme is to select a binary code matrix ${C}$ to determine the local decision rules at the nodes, and to perform fault-tolerant fusion at the FC. For a network with $N$ nodes trying to distinguish among $M$ hypotheses, the code matrix ${C}$ is an $M\times N$ binary matrix. Each row of ${C}$ corresponds to one of the $M$ possible hypotheses $H_1,\cdots,H_M$ and each column represents the binary decision rule of the corresponding node. Given this code matrix, the node $j$ sends its binary decision $u_j\in\{0,1\}$ to the FC. After receiving the binary decisions $\mathbf{u}=(u_1,\cdots,u_N)$ from local nodes, the final classification decision is made at the FC using minimum Hamming distance based fusion given by:

Decide $H_m$ where 
\begin{equation}
\label{eq:hamm}
m=\argmin_{1 \leq l\leq M}{ d_H(\mathbf{u},\mathbf{r}_l)},
\end{equation}
where $d_H(\mathbf{x},\mathbf{y})$ is the Hamming distance between $\mathbf{x}$ and $\mathbf{y}$, and $\mathbf{r}_l=(c_{l1},\cdots,c_{lN})$ is the $l$th row of ${C}$ which corresponds to hypothesis $H_l$.
The tie-break rule is to randomly pick a row of the code matrix ${C}$ from those with the smallest Hamming distance to the received vector $\mathbf{u}$. The performance of the scheme depends on the code matrix ${C}$ since it is used for designing the local decision rules as well as for the fusion rule at the FC. Several approaches to design the matrix $C$, e.g., based on simulated annealing and cyclic column replacement, were presented in~\cite{Wang_jsac05}. 

For example, consider the code matrix used by a parallel network of $N=7$ nodes performing an $(M=4)$-ary classification problem   
\[ C= \left[ \begin{array}{ccccccc}
1 & 0 & 0 & 0 & 1 & 0 & 1 \\
0 & 0 & 1 & 0 & 0 & 0 & 0 \\
1 & 0 & 1 & 1 & 0 & 1 & 0 \\
0 & 1 & 1 & 1 & 1 & 1 & 1  \end{array} \right].\]
When the true hypothesis is $H_1$ corresponding to the first row, all the nodes are supposed to send the first element of their column. However, due to imperfect observations at the nodes, consider the case when the FC receives the vector $[1 1 1 0 1 0 1]$. The FC evaluates the Hamming distance between this received vector and each of the rows resulting in the Hamming distance values $(2,4,5,3)$. Therefore, it decides the hypothesis corresponding to the first row, $H_1$, as the true hypothesis.

\section{Distributed Classification in Tree Networks}
\label{sec:approachclass}
In this section, we consider the problem of distributed classification in tree networks. 
We model the classification problem as an $M$-ary hypotheses testing problem. Let $H_l$, where $l=1,\cdots,M$ and $M\geq 2$, denote the $M$ hypotheses\footnote{In order to distinguish among the $M$ hypotheses using binary decisions, we assume that $N\geq{\log}_2 M$.}. The \textit{a priori} probabilities of these $M$ hypotheses are denoted by $Pr(H_l)=P_l$, for $l=1,\cdots,M$. 

\subsection{Proposed Scheme}
\label{ap:class}
We assume that under each hypothesis $H_l$, every leaf node $j$ acts as a source and makes an independent and identically distributed (i.i.d.) observation $y_j^K$. After processing the observations locally, every leaf node $j$ sends its local decision\footnote{In this context, ``decision" is a binary quantized value determined by the processing model.} $u_{j}^{K}\in \{0,1\}$ according to a transmission mapping $\tau_j^K(\cdot)$ to its immediate predecessor $P^{K-1}(j)$. 
Each intermediate node $j$ at level $k$ receives the decision vector consisting of local decisions made by its immediate successors $S^{k+1}(j)$ at level $k+1$, which can be expressed as $\mathbf{v_j^k}=\mathbf{u^{k+1}}=(u_{1}^{k+1},\cdots,u_{N}^{k+1})$. After fusing this data using fusion rule $f_j^{k}(\cdot)$, this intermediate node $j$ at level $k$ makes a classification decision $y_j^k\in\{1,\cdots,M\}$. Then, it sends a 1-bit version of this decision, $u_{j}^{k}\in \{0,1\}$, according to its transmission mapping $\tau_j^{k}(\cdot)$ to its immediate predecessor $P^{k-1}(j)$. Finally, the FC receives the decision vector $\mathbf{u^1}=(u_{1}^{1},\cdots,u_{N}^{1})$ and fuses this data to decide the underlying hypothesis.
The proposed scheme builds on the DCFECC scheme (Section~\ref{sec:DCFECC}). To summarize, each node $j$ at level $k$, for $1\leq k\leq K-1$, performs two basic operations (Please see Figure~\ref{syst2}): 
\begin{itemize}
\item Collect data from its successors $S^{k+1}(j)$ and fuse their data using fusion rule $f_j^k(\cdot)$ to locally decide the hypothesis, denoted by $y_j^k$.
\item Compress the decision $y_j^k$ at node $j$ about the hypothesis and transmit a 1-bit version $u_j^k$ to its predecessor $P^{k-1}(j)$ using the transmission mapping $\tau_j^k(\cdot)$. 
\end{itemize}

\begin{figure}[t]
  \centering
    \includegraphics[height=1.5in, width=!]{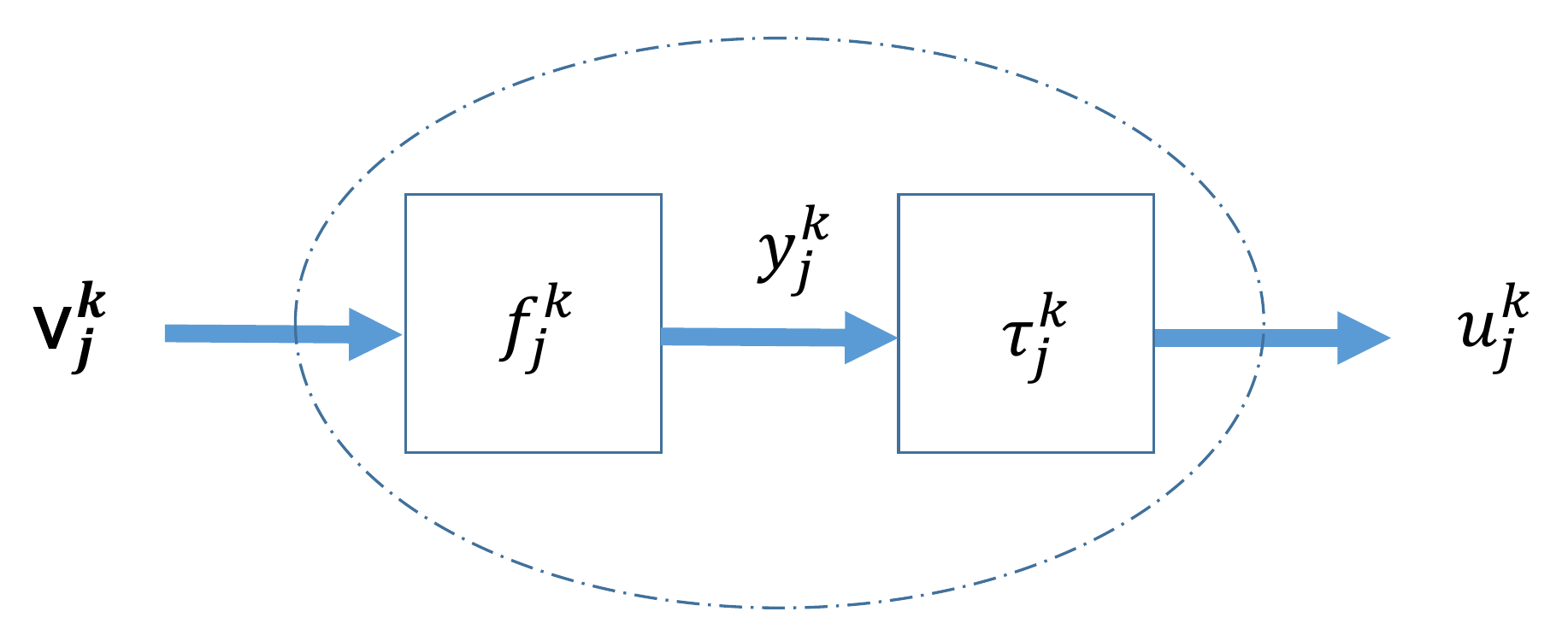}
    \vspace{-0.2in}
    \caption{Data processing for distributed classification at node $j$ at level $1\leq k\leq K-1$. Here $\mathbf{v_j^k} \in \{0,1\}^N$, $y_j^k\in\{1,\cdots,M\}$, and $u_j^k\in\{0,1\}$. Therefore, the mappings are $f_j^k: \{0,1\}^N \to \{1,\cdots,M\}$ and $\tau_j^k: \{1,\cdots,M\} \to \{0,1\}$}\label{syst2}
\end{figure}

For the leaf nodes (level $K$), there are no successors and, therefore, only the second operation needs to be performed. And, for the FC (level `0'), only the first operation needs to be performed. Each of the functions $f_j^k(\cdot)$ and $\tau_j^k(\cdot)$ depend on the code matrix used at level $k$. By appealing to symmetry, we assume that each node at the same level $k$, uses an identical code matrix $C^k$ for transmission to its predecessor and $C^{k+1}$ for fusion of data from its successors. We start with the design of transmission mapping $\tau^K_j(\cdot)$ of the leaf nodes. When the leaf nodes use code matrix $C^K$ for transmission, the probability of misclassification at level $K-1$ is given by \cite{Wang_jsac05}

\begin{equation}
P_e^{K-1}= 
\sum_{\mathbf{i},l}\int_{\mathbf{y^K}}P_lP(u_1^K=i_1|y_1^K)\times\cdots\times P(u_N^K=i_N|y_N^K)p(\mathbf{y^K}|H_l) \psi^{K}_{\mathbf{i},l},\label{eq:leaf}
\end{equation}

where $\mathbf{i}=[i_1,\cdots,i_N]\in\{0,1\}^N$ is a realization of the received codeword $\mathbf{u^K}$, $\mathbf{y^K}=[y^K_1,\cdots,y^K_N]$ are the local observations of leaf nodes, and $\psi^{K}_{\mathbf{i},l}$ is the cost associated with a global decision $H_l$ at level $K-1$ when the received vector from level $K$ is $\mathbf{i}$. This cost is:
\begin{equation}
\label{eq:cost}
\psi_{\mathbf{i},l}^k=
\begin{cases}
1-\frac{1}{\varrho} &\mbox{if } \mathbf{i} \mbox{ is in the decision region of } H_l\\
1 & \mbox{otherwise.}
\end{cases} 
\end{equation}
for $k=K$, where $\varrho$ is the number of decision regions corresponding to a received codeword $\mathbf{i}$. In other words, it is the number of rows of code matrix $C^K$ which have the same minimum Hamming distance with the received codeword $\mathbf{i}$. Usually this value is 1, however $\varrho$ can be greater than one when there is a tie at the node at level $K-1$ and in those cases, the tie-breaking rule is to choose one of them randomly.

Employing a person-by-person optimization approach, we can find the local transmission mapping of the leaf nodes as follows \cite{Wang_jsac05}:
\begin{equation}
u_j^K=\tau_j^K(y_j^K)=\\
\begin{cases}
0,	&\text{if $\sum_{l}p(y_j^K|H_l)A_{jl}<0$}\\
1,	&\text{otherwise}
\end{cases},
\end{equation}
where $A=\{A_{jl}\}$ is the weight matrix whose values\footnote{We refer the reader to \cite{Wang_jsac05} for further details.} are given by,

\begin{multline}
A_{jl}=\sum_{i_1,\cdots,i_{j-1},i_{j+1},\cdots,i_N}P_lP(u_1^K=i_1|H_l) 
\times\cdots\times P(u^K_{j-1}=i_{j-1}|H_l)P(u^K_{j+1}=i_{j+1}|H_l)\\
\times\cdots\times P(u_{N}^K=i_{N}|H_l) 
\times [\psi^K_{i_1,\cdots,i_{j-1},0,i_{j+1},\cdots,i_N,l}-\psi^K_{i_1,\cdots,i_{j-1},1,i_{j+1},\cdots,i_N,l}].
\end{multline}

For $1\leq k\leq K-1$, the local classification decision $y_j^k \in \{1,\cdots,M\}$ made using the data from the successors $S^{k+1}(j)$ is discrete and, therefore, the transmission mapping $\tau^k(\cdot)$ is straight-forward and is given as follows:
\begin{equation}
u_j^k=\tau^k_j(y_j^k)=c^{k}_{y_j^kj},	\quad\text{if $1\leq k\leq K-1$.}
\end{equation}

In other words, the one-bit decision $u_j^k$ is the element of $C^{k}$ corresponding to $y_j^k$th row and $j$th column. For every intermediate node, the fusion rule $f_j^k(\cdot)$ is the minimum Hamming distance fusion rule as given in \eqref{eq:hamm}. Therefore, the performance of the scheme depends on the minimum Hamming distance of the code matrices. Let $d_{min}^k$ be the minimum Hamming distance of the code matrix $C^k$. In the remainder of this section, we derive the error expressions at intermediate nodes which will later be used for the design of code matrices at every level. 

\begin{proposition}
\label{prop1}
The probability of misclassification $P_e^{k-1}$ at level $k-1$ due to the data received from level $k$ and using code matrix $C^k=\{c_{mj}^k\}$ ($1\leq k\leq K-1$, $1\leq m\leq M$, $1\leq j\leq N$) is:

\begin{equation}
\label{eq:pe_inter}
P_e^{k-1}=\sum_{\mathbf{i},l}P_l\prod_{j=1}^N \left[(2i_j-1)\sum_{m=1}^Mc^k_{mj}P^k_{ml} +(1-i_j)\right]\psi^k_{\mathbf{i},l} \mbox{,}
\end{equation}

where $\mathbf{i}=[i_1, \cdots, i_N] \in \{0, 1\}^N$ is the realization of the received codeword $\mathbf{u^k}$, matrix $P^k=\{P_{ml}^k\}$ is the confusion matrix of the local decisions at level $k$, and $\psi^k_{\mathbf{i},l}$ is the cost associated with a global decision $H_l$ at level $k-1$ when the received vector from level $k$ is $\mathbf{i}$. This cost is given by \eqref{eq:cost}.
\end{proposition}
\begin{IEEEproof}

If $u_j^k$ denotes the bit sent by the node $j$ at level $k$ and the global decision is made using the Hamming distance criterion:
\begin{equation}
P_{e}^{k-1}=\sum_{\mathbf{i},l}P_lP(\mathbf{u^k}=\mathbf{i}|H_l)\psi_{\mathbf{i},l}^k\mbox{.}
\end{equation}
Since local decisions are conditionally independent, $P(\mathbf{u^k}=\mathbf{i}|H_l)=\prod_{j=1}^N P(u_j^k=i_j|H_l)$.
Further,

\begin{align*}
P(u_j^k=i_j|H_l) &= i_jP(u_j^k=1|H_l)+(1-i_j)P(u_j^k=0|H_l)\\
&=(1-i_j)+(2i_j-1)P(u_j^k=1|H_l)\\
&=(1-i_j)+(2i_j-1)\sum_{m=1}^Mc^k_{mj}P(y_j^k=m|H_l)\\
&=(1-i_j)+(2i_j-1)\sum_{m=1}^Mc^k_{mj}P^k_{ml}
\end{align*}

where $y_j^k$ is the local classification decision made by node $j$ after collecting data from its successors $S^{k+1}(j)$ at level $k+1$. The desired result follows.
\end{IEEEproof}

Note that this suggests that the probability of misclassification at level $k-1$ is dependent on the confusion matrix at level $k$. These can be derived easily as follows:
\begin{equation}
P^k_{ml}\defeq P(\text{decide $H_m$ at level $k$}|\text{$H_l$ is true})= 
1-\sum_{\mathbf{i}}p(\mathbf{u^k}=\mathbf{i}|H_l)\psi_{\mathbf{i},m}^{k+1}\label{eq:conf}
\end{equation}

From these expressions, we can observe that there is a recursive structure, where the probability of misclassification at level $k$ is dependent on the confusion matrix of level $k+1$. Therefore, the performance at the FC depends on all the code matrices in a recursive manner. As mentioned before, we propose a simpler approach by assuming that each node of the same level uses the same code matrix which is designed by optimizing on a person-by-person sequential basis. We start with the code design at level $K-1$ to fuse data from level $K$. This is designed by optimizing the expression in \eqref{eq:leaf}. Once we have designed the optimal code matrix at this level, we derive the corresponding confusion matrix from \eqref{eq:conf}, which is used to design the code matrix at the next level by optimizing the expression in \eqref{eq:pe_inter}. Following this method, we can design all the code matrices. Note that each of these optimizations can be performed offline using approaches such as simulated annealing or cyclic-column replacement \cite{Wang_jsac05}.

In the following subsection, we analyze our scheme in the asymptotic regime and show that the scheme is asymptotically optimal. 

\subsection{Asymptotic Optimality}
\label{sec:a-optimalclass}
We study the asymptotic classification performance of our scheme for two different classes
of tree networks. The first one is the class of \textit{fixed height trees} in which the
height of the tree, $K$, is assumed to be fixed while the second is the class of \textit{fixed degree trees} in which the degree of the tree, $N$, is assumed to be fixed.
More specifically, we study the classification performance of minimum Hamming distance fusion in fixed height tree networks when the number of nodes tends to infinity and in fixed degree tree networks when the height of the tree tends to infinity.
We first provide the following bound on the misclassification probability at the FC which will be used to prove the asymptotic optimality. Let $Q^k_{m}$ be the probability of misclassifying hypothesis $H_m$ at level $k$ and define $q_{max}^k \defeq \max_{1\leq m\leq M} Q_m^k$. Note that for levels $0\leq k \leq K-1$, we have $Q^k_{m}=1-P_{mm}^k$ where $P_{ml}^k$ are the elements of the confusion matrix. For $k=K$, $Q_m^K=1-Pr(\text{decide $H_m$ at level $K | H_m$ is true})$.

\begin{proposition}
\label{asy}
In a perfect tree structure $T(K,N)$ employing the proposed scheme, the misclassification probability at the FC, $P_{e}^{0}$, is bounded as follows
\begin{equation}
P_{e}^{0}\leq \left[q_{max}^K\right]^{\displaystyle\prod_{k=1}^{K}\frac{d_{min}^k}{a_k}},\label{eq:bound}
\end{equation}
if $q_{max}^K< \frac{1}{2}$ and 
\begin{equation}
\label{dcon}
d_{min}^k\geq \dfrac{2(M-2)}{[1-4q_{max}^k(1-q_{max}^k)]-(1/a_k)[(2/q_{max}^k)-2]},\;\forall k,
\end{equation}
where $a_k$ is a parameter which satisfies the following condition
\begin{equation}
\label{acon}
a_k > \dfrac{2(1-q_{max}^k)}{q_{max}^k-4(q_{max}^k)^2(1-q_{max}^k)},\forall k.
\end{equation}
\end{proposition}

\begin{proof}
To prove the proposition, we start with the inequality \eqref{dcon}

\begin{eqnarray}
d_{min}^k\geq \dfrac{2(M-2)}{[1-4q_{max}^k(1-q_{max}^k)]-(1/a_k)[(2/q_{max}^k)-2]}\\
\implies\frac{d_{min}^k}{2}\left(1-4q_{max}^k(1-q_{max}^k)\right)\geq (M-2)+\frac{d_{min}^k}{a_k}\left(\frac{1}{q_{max}^k}-1\right)\label{eq11}\\
\implies\frac{d_{min}^k}{2}\log\left(\frac{1}{4q_{max}^k(1-q_{max}^k}\right)\geq \log(M-1)+\frac{d_{min}^k}{a_k}\log\frac{1}{q_{max}^k}\label{thirdineq}\\
\implies\left[q_{max}^k\right]^{\frac{d_{min}^k}{a_k}}\geq (M-1) \left[\sqrt{4q_{max}^{k}(1-q_{max}^{k})}\right]^{d_{min}^k}
\end{eqnarray}

where \eqref{eq11} is true because $a_k$ satisfies \eqref{acon}, and \eqref{thirdineq} can be proved by applying the logarithm inequality: $(x-1)\geq\log x \geq \dfrac{x-1}{x}$, for $x>0$. Now, for $k=1,\cdots,K$

\begin{eqnarray}
Q_{m}^{k-1}= Pr(\text{decision at level $k-1$} \neq H_m\;|\; H_m)&\leq&  Pr(d^k(\mathbf{u^k},\mathbf{c_m^k})\geq \underset{1 \leq l \leq M,\;l\neq m}{\min}d^k(\mathbf{u^k},\mathbf{c_l^k})\;|\;H_m)\nn\\
&\leq& \sum\limits_{\underset{l\neq m}{l=1}}^{M} Pr(d^k(\mathbf{u^k},\mathbf{c_m^k})\geq d^k(\mathbf{u^k},\mathbf{c_l^k})\;|\;H_m)\nn\\
&\leq& \sum\limits_{\underset{l\neq m}{l=1}}^{M} \left[\sqrt{4Q_{mm}^{k}(1-Q_{mm}^{k})}\right]^{d^k(\mathbf{c_l^k},\mathbf{c_m^k})}\label{17}\\
&\leq& (M-1) \left[\sqrt{4q_{max}^{k}(1-q_{max}^{k})}\right]^{d_{min}^k}\\
&\leq& \left[q_{max}^k\right]^{\frac{d_{min}^k}{a_k}}\label{19}.
\end{eqnarray}

Note that \eqref{17} is true when $Q_{m}^{k}<\frac{1}{2}$ as shown in \cite{ChenWHKY05}, which holds when \eqref{dcon} and \eqref{acon} are true. Using the above results, the average probability of error can be bounded as follows.
\begin{eqnarray*}
P_e^0&=& \sum\limits_{m=1}^{M}P_m Pr(\text{decision at the FC}\neq H_m\;|\; H_m)\\
&\leq&\sum\limits_{m=1}^{M}P_m \left[q_{max}^1\right]^{\frac{d_{min}^1}{a_1}}\\
&=& \left[q_{max}^1\right]^{\frac{d_{min}^1}{a_1}}
\end{eqnarray*}
Now, since $Q_{m}^{1}\leq\left[q_{max}^2\right]^{\frac{d_{min}^2}{a_2}}$ $\forall m$, we have $q_{max}^{1}\leq\left[q_{max}^2\right]^{\frac{d_{min}^2}{a_2}}$. Continuing in this manner, we get
\begin{eqnarray*}
P_e^0&\leq& \left[q_{max}^K\right]^{\prod_{k=1}^{K}\frac{d_{min}^k}{a_k}}.
\end{eqnarray*}
\end{proof}

The results obtained in Proposition~\ref{asy} show that the misclassification probability for minimum Hamming distance fusion can be upper-bounded by a quantity determined by the minimum Hamming distance of the code matrices ($d_{min}^k,\;\forall k$), and the largest local classification error among all hypotheses ($q_{max}^k,\;\forall k$). Also, note that the parameter $a_k$ in \eqref{eq:bound} can be chosen appropriately to make the bound tighter. For example, if $a_k$ is chosen such that

\begin{equation*}
\dfrac{q_{max}^k(2M-2)+2(1-q_{max}^k)}{q_{max}^k-4(q_{max}^k)^2(1-q_{max}^k)}>a_k > \dfrac{2(1-q_{max}^k)}{q_{max}^k-4(q_{max}^k)^2(1-q_{max}^k)},\forall k,
\end{equation*}

then, $(d_k/a_k)>1,\forall k$, and we have $$P_{e}^{0}\leq \left[q_{max}^K\right]^{\prod_{k=1}^{K}(d_{min}^k/a_k)}\leq \left[q_{max}^K\right]^{(d_{min}^K/a_K)}.$$ 
As a consequence, for fixed height trees, the decoding error of the proposed scheme vanishes as $d_{min}^K$ approaches infinity which happens when $N\to\infty$. 
Also, for fixed degree trees, 
the decoding error of the proposed scheme vanishes as $K$ approaches infinity.
These results can be summarized in the following theorem:

\begin{theorem}
Under conditions \eqref{dcon} and \eqref{acon}, the proposed coding theory based distributed classification scheme is asymptotically optimal, for both classes of tree networks: fixed
height tree networks and fixed degree tree networks, as long as the probabilities of correct local
classification for all hypotheses of the leaf nodes are greater than one half. 
\end{theorem}

The conditions required for the above theorem depend on the minimum Hamming distance $d_{min}^k$ of the code matrices used at each level and can be interpreted as follows: the proposed scheme is optimal when the minimum Hamming distance of the code matrices is ``large enough" to ensure that perfect classification is made at every level of the tree. When the rows of the code matrices are well separated due to
large minimum Hamming distance, the proposed scheme can handle more errors and have good performance.

Also, observe that these results imply that when \eqref{dcon} and \eqref{acon} are satisfied, there is no loss in asymptotic performance when all the nodes at level $k$, for $k=1,\cdots,K$, use identical transmission mapping and identical fusion rules.

\subsection{Simulation Results}
\label{sec:simclass}
In this section, we evaluate the performance of the proposed scheme using simulations. Consider a tree network $T(3,7)$ consisting of a total $N_{total}=400$ nodes, including the FC. The leaf nodes sense the environment to identify among four ($M=4$) equally likely hypotheses. As discussed before, we assume that all the leaf node measurements are independent and identically distributed. Under each hypothesis, the probability density function is assumed to
be Gaussian distribution with the same variance ($\sigma^2=1$) but with different means $0,s,2s$, and $3s$ respectively. The signal-to-noise power ratio (SNR) of observations at each local node is given by $20\log_2{s}$. The code matrices are designed using the scheme described in Section~\ref{ap:class} and simulated annealing for optimization. The designed code matrices used at different levels of the tree are found to be

\begin{equation}
C^1=[11,8,9,9,3,9,12]
\end{equation}
\begin{equation}
C^2=[7,6,3,12,12,9,14]
\end{equation}
\begin{equation}
C^3=[3,8,14,12,9,12,9]
\end{equation}

where the code matrix is represented by a vector of $M$-bit integers. Each integer $m_j$ represents a
column of any arbitrary code matrix $C$ and can be expressed as $m_j=\sum_{l=1}^{M}c_{lj}$. For example, the integer $9$ in column 5 of $C^3$ represents $c^3_{15}=1,c^3_{25}=0,c^3_{35}=0$, and $c^3_{45}=1$.

\begin{figure}[t]
\centering
\includegraphics[width=3in,height=!]{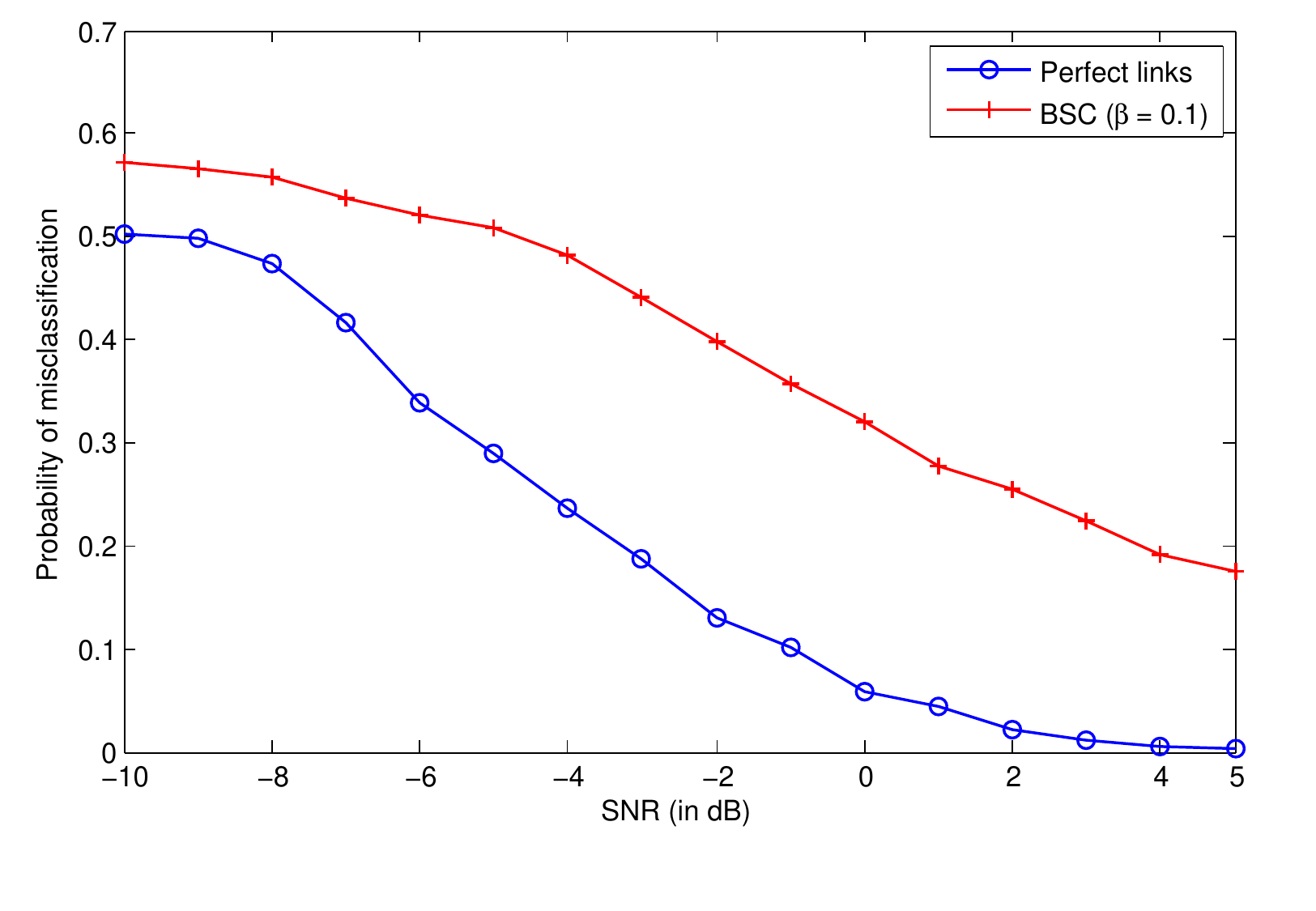}\vspace{-0.25in}
\caption{Probability of misclassification versus SNR}
\label{fig:sim}
\end{figure}

In Figure~\ref{fig:sim}, we plot the final probability of misclassification at the FC with varying SNR values. Note that this probability of misclassification is empirically found by performing $N_{mc}=5000$ Monte-Carlo runs. As we can observe, the performance of the scheme improves with increasing SNR and approaches 0 as early as 5dB. Since the proposed scheme is based on error-correcting codes, it can also tolerate some errors in data. These errors could be due to various reasons: presence of a faulty node \cite{Wang_jsac05}, presence of imperfect links between levels \cite{varshney_spmag06}, or presence of a malicious node sending falsified data \cite{vempaty_spm13}. In order to check the fault-tolerance capability of the scheme, we have simulated the case when the links between the levels are binary symmetric channels with crossover probabilities $\beta=0.05$ and $\beta=0.1$. As shown in Figure~\ref{fig:sim}, the proposed scheme still performs reasonably well even in the presence of imperfect data due to non-ideal channels modeled as binary symmetric channels.

Building on these results, in the following section, we address the parameter estimation problem in tree based networks. More specifically, we break the parameter estimation problem into a sequence of $M$-ary decision making problems, and each of these $M$-ary decision making problems is solved using a technique similar to the distributed classification scheme of the previous section. 

\section{Distributed Parameter Estimation in Tree Networks using Iterative Classification}
\label{sec:approachest}
Consider a distributed parameter estimation problem where the goal is to estimate a random scalar parameter $\theta$ at the FC. The parameter $\theta$ has a prior probability density function (pdf) $p_\theta(\theta)$ where $\theta\in \Theta$. We propose a scheme to estimate the parameter $\theta$ using iterative classification. By doing
so, we break the parameter estimation problem into a sequence of $M$-ary decision making problems. This is essentially a process of iterative rejection of unlikely objects where the most undesirable options are discarded and the scope of options is progressively narrowed down until exactly one option is left. 

\subsection{Proposed Scheme}
\label{sec:systestm}
\begin{figure}[t]
\centering
\includegraphics[width=3in,height=!]{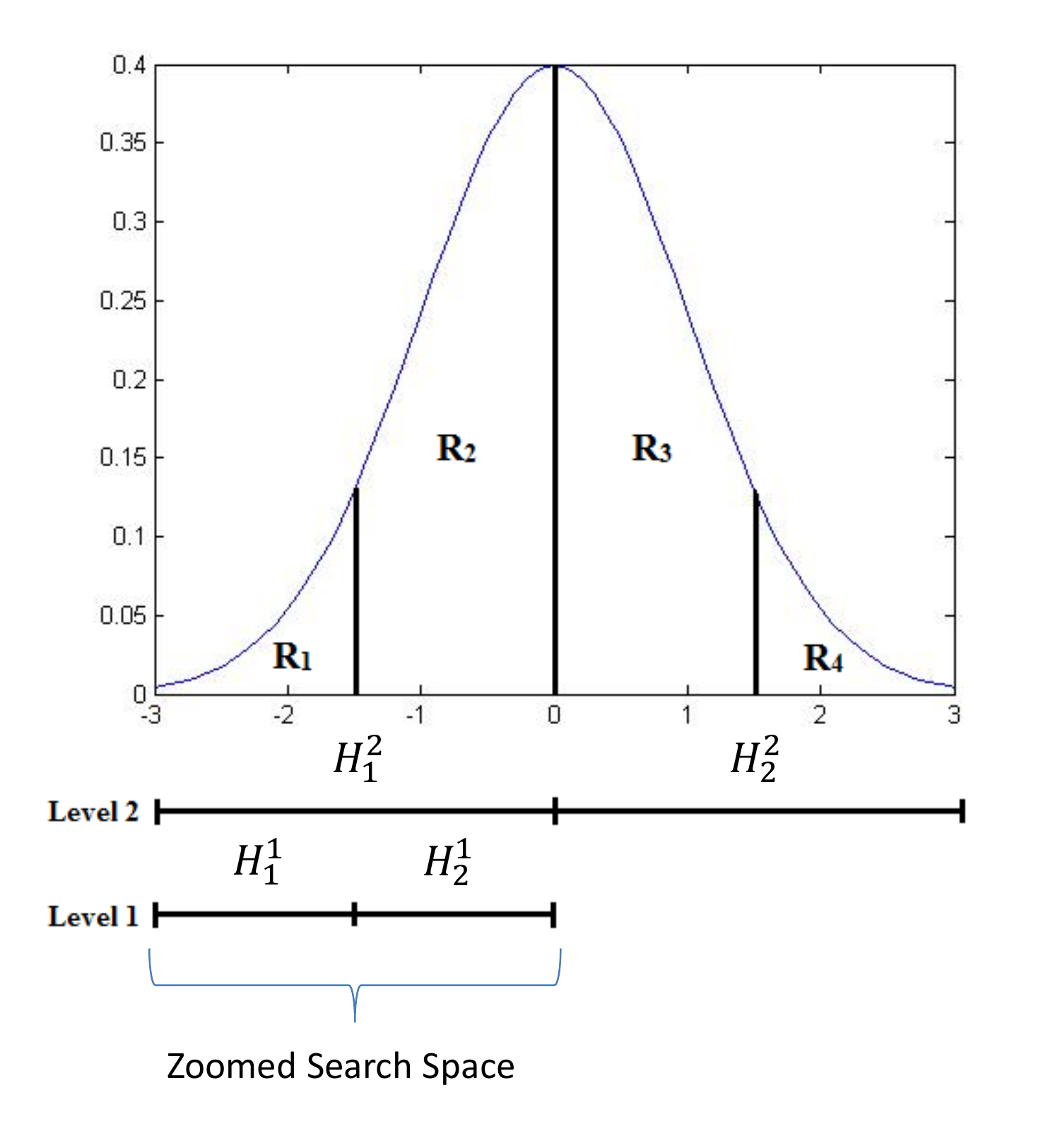}\vspace{-0.25in}
\caption{An example of splitting of parameter space.}
\label{fig:normal}
\end{figure}

We consider a distributed estimation system with the topology of a perfect tree, $T(K,\;N)$, rooted at the FC.
We model the parameter estimation problem as an $M$-ary hypotheses testing problem. Our scheme is iterative in which at every iteration $1\leq s \leq K$, the parameter space is split into $M$ regions and an $M$-ary hypothesis test is performed at the level $(K+1-s)$ of the tree to determine the parameter space for the next level in the tree. The optimal splitting of the parameter space at every iteration can be determined offline (which will be explained later in the paper in Section~\ref{sec:split}). For now, we assume that the $M^K$ final regions and their corresponding representation points are known. Let $H_l^k$, where $l=1,\cdots,M$ and $M\geq 2$, denote the $M$ hypotheses\footnote{As before, we assume that $N\geq{\log}_2 M$.} being tested at level $k$. Figure~\ref{fig:normal} shows an example of parameter space splitting when $p_\theta(\cdot)$ is standard Gaussian, and $M=K=2$. Every node at level $k=2$ first performs a classification task to determine if the parameter $\theta$ is positive or negative (differentiate between hypotheses $H_1^2$ and $H_2^2$). After a decision is made, the nodes at level $k=1$, `zoom' into the decided hypothesis, say $H_1^2$, and perform a classification task to determine if $\theta$ belongs to hypothesis $H_1^2$ or $H_2^1$. In this manner, the FC at level `0' eventually decides the true hypothesis, among the $M^K$ hypotheses, where $\theta$ belongs.

The \textit{a priori} probabilities of the $M^K$ hypotheses are denoted by $Pr(H_l^k)=P_l^k$, for $l=1,\cdots,M$ and $k=1,\cdots,K$. $P_l^k$ depends on $p_\theta(\cdot)$ and the region corresponding to $H_l^k$.
We assume that every node $j'$ at level $k+1$ acts as a source and makes a conditionally independent and identically distributed (i.i.d.) observation $y_{j'}^{k+1}$, conditioned under each hypothesis $H_l^k$. After processing the observations locally, every node sends its local decision $u_{j'}^{k+1}\in \{0,1\}$ according to a transmission mapping $\tau^{k+1}(\cdot)$ to its immediate predecessor $P^k(j')$. 
Each intermediate node $j$ at level $k$ receives the decision vector $\mathbf{v_{j}^k}$ consisting of local decisions made by its immediate successors $S^{k+1}(j)$ at level $k+1$. Intermediate nodes at level $k$, through collaboration\footnote{In collaboration phase, node $j$ at level $k$ shares $\mathbf{v_j^k}$ (the data collected from its successors $S^{k+1}(j)$) with other nodes at level $k$. In this paper, we assume that nodes do not compress $\mathbf{v_j^k}$ for collaboration and, therefore, after collaboration phase receive the data $\mathbf{v^k}=\left[\mathbf{v_1^k},\cdots,\mathbf{v_N^{k}}\right]\in\{0,1\}^{N^k}$.} and fusion, decide on the result of the $M$-ary hypotheses test as the new parameter space for them. 
\begin{figure}[t]
  \centering
    \includegraphics[height=.95in, width=!]{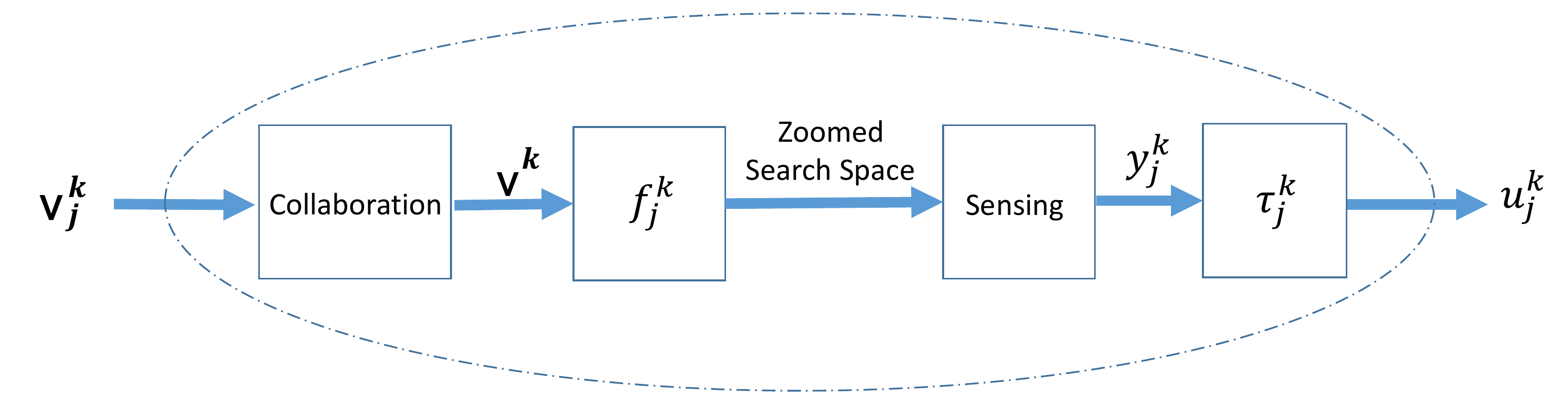}
    \vspace{-0.2in}
    \caption{Data processing for distributed estimation at node $j$ at level $1\leq k\leq K-1$. Here, $\mathbf{v_j^k}\in\{0,1\}^N$, $\mathbf{v^k}\in\{0,1\}^{N^k}$, $y_j^k\in\mathbb{R}$, and $u_j^k\in\{0,1\}$. Therefore, the mappings are $f_j^k:\{0,1\}^{N^k}\to\{1,\cdots,M\}$ and $\tau_j^k:\mathbb{R}\to\{0,1\}$.}\label{syst3}
\end{figure}

The scheme builds on the DCFECC scheme proposed for distributed classification. Each node $j$ at level $k$, for $1\leq k\leq K-1$ performs four basic operations (Please see Figure~\ref{syst3}): 
\begin{itemize}
\item Collect data from its successors $S^{k+1}(j)$ and collaborate with other nodes at level $k$.
\item Decide the new parameter space (hypotheses to test) by fusing data using fusion rule $f_j^k(\cdot)$. 
\item Acquire observation $y_j^k$ and perform hypothesis testing to determine the new parameter space.
\item Compress the observation at node $j$ about the hypothesis (new parameter space) and transmit a 1-bit version to the predecessor $P^{k-1}(j)$ using the transmission mapping $\tau_j^k(\cdot)$.
\end{itemize}

For the leaf nodes (level $K$), there are no successors and, therefore, only the third and fourth operations need to be performed. The FC (level `0') collects data from its successors and makes the final decision regarding the region where $\theta$ belongs. Given a tree network, our objective is to find efficient transmission
mappings and fusion rules for nodes at all levels, to
maximize the estimation performance at the FC.

\textit{Remark:} There are three major differences between the scheme proposed here for distributed estimation in tree networks and the scheme proposed in Section~\ref{sec:approachclass} for distributed classification in tree networks:
\begin{itemize}
\item In the scheme proposed here, every node acts as a source node and senses the phenomenon while for the classification problem in Section~\ref{sec:approachclass}, only the leaf nodes act as source nodes and intermediate nodes act only as relay nodes.
\item In Section~\ref{sec:approachclass}, each node performs the same classification task, or in other words, the set of classes are the same. On the other hand, in the scheme proposed here, the parameter space is `zoomed' at every level, which changes the corresponding classes to be tested.
\item An important step in the scheme proposed in this section is the collaboration step which is not required for the classification problem of Section~\ref{sec:approachclass}.
\end{itemize}

By appealing to symmetry, we assume that each node at the same level $k$, uses an identical code matrix $C^k$ for transmission to its predecessor and $C^{k+1}$ for fusion of data from its successors. Each of the functions $f_j^k(\cdot)$ and $\tau_j^k(\cdot)$ depend on the code matrix used at level $k$. Although the performance metric in this
framework is the Mean Square Error (MSE), it is difficult
to obtain a closed form representation for MSE. Therefore,
typically, one uses the bounds on MSE to characterize the
performance of the estimator. Here, we use an analytically tractable metric to analyze the performance of the proposed scheme which is the probability of misclassification of the parameter region. It is an important metric when the final goal of the parameter estimation task is to find the approximate region or neighborhood where the parameter lies rather than the true value of the parameter itself. Since the final region could be one of the $M^K$ regions, a metric of interest is the probability of `zooming' into the correct region. In other words, it is the probability that the true value of the parameter and the estimated value of the parameter lie in the same region.

Now, we design the transmission mapping $\tau_j^k(\cdot)$ of nodes at level $k$. Notice that the final region of the estimated value of the parameter is the same as the true value of the parameter, if and only if we `zoom' into the
correct region at every iteration of the proposed scheme. Thus, when the nodes at level $k$ use code matrix $C^k$ for transmission, the probability of misclassification at level $k-1$ is given by the following proposition.

\begin{proposition}
\label{prop2}
The probability of misclassification $P_e^{k-1}$ at level $k-1$ due to the data received from level $k$ and using code matrix $C^k=\{c_{mj}^k\}$ ($1\leq k\leq K$, $1\leq m\leq M$, $1\leq j\leq N^k$) is

\begin{equation}
P_e^{k-1}= 1-\prod\limits_{t=k}^{K}
\Bigg[1-\sum_{\mathbf{i},l}\int_{\mathbf{y^t}}P_l^t P(u_1^t=i_1|y_1^t)\times\cdots 
\times P(u_{N^t}^t=i_{N^t}|y_{N^t}^t)p(\mathbf{y^t}|H_l^t) \psi^{t}_{\mathbf{i},l}\Bigg],\label{eq:nk}
\end{equation}

where $\mathbf{i}=[i_1, \cdots, i_{N^t}] \in \{0, 1\}^{N^t}$ is the realization of the received codeword $\mathbf{u^t}$, $\mathbf{y^t}=[y^t_1,\cdots,y^t_{N^t}]$ are the local observations of nodes at level $t$, matrix $P^t=\{P_{ml}^t\}$ is the confusion matrix of the local decisions at level $t$, and $\psi^{t-1}_{\mathbf{i},l}$ is the cost associated with a global decision $H_l^{t-1}$ at level $t-1$ when the received vector from level $t$ is $\mathbf{i}$. This cost is:
\begin{equation}
\label{eq:cost1}
\psi_{\mathbf{i},l}^t=
\begin{cases}
1-\frac{1}{\varrho} &\mbox{if } \mathbf{i} \mbox{ is in decision region of } H_l^t\\
1 & \mbox{otherwise.}
\end{cases} 
\end{equation}
where as before $\varrho$ is the number of decision regions corresponding to a received codeword $\mathbf{i}$. In other words, it is the number of rows of code matrix $C^t$ which have the same minimum Hamming distance with the received codeword $\mathbf{i}$. 
%Typically this value is 1, however $\varrho$ can be greater than one when there is a tie at the node at level $t-1$ and in those cases, the tie-breaking rule is to choose one of them randomly.
$P_l^t$ is the prior probability of hypothesis $H_l^t$ at level $t$. 

\end{proposition}
\begin{proof}
Note that a correct decision is made at level $k-1$ if and only if the decision at all levels from $t=k$ to $t=K$ are correct. Therefore, using \eqref{eq:leaf} in a recursive manner at every level of the tree, we get the desired result.
\end{proof}
From \eqref{eq:nk}, we can observe that the performance at the FC depends on all the code matrices in a recursive manner. In this paper, we employ a simpler approach by assuming that code matrices are designed by optimizing on a person-by-person basis. The code matrix at each level is designed using an approach similar to that of a parallel topology. Note that each of these optimizations can be performed offline using approaches such as simulated annealing or cyclic-column replacement \cite{Wang_jsac05}.

Employing a person-by-person optimization approach, we can find the local transmission mapping of the nodes at level $k$ as follows:
\begin{equation}
u_j^k=\tau_j^k(y_j^k)=\\
\begin{cases}
0,	&\text{if $\sum_{l}p(y_j^k|H_l^k)A_{jl}^k<0$}\\
1,	&\text{otherwise}
\end{cases},
\end{equation}
where $A^k=\{A^k_{jl}\}$ is a weight matrix whose values are given by,

\begin{multline}
A_{jl}^k=\sum_{i_1,\cdots,i_{j-1},i_{j+1},\cdots,i_{N^k}}P_l^k P(u_1^k=i_1|H_l^k) 
\times\cdots\times P(u^k_{j-1}=i_{j-1}|H_l^k)P(u^k_{j+1}=i_{j+1}|H_l^k)\\
\times\cdots\times P(u_{N^k}^K=i_{N^k}|H_l^k) 
\times [\psi^k_{i_1,\cdots,i_{j-1},0,i_{j+1},\cdots,i_{N^k},l}-\psi^k_{i_1,\cdots,i_{j-1},1,i_{j+1},\cdots,i_{N^k},l}].
\end{multline}

For every intermediate node, the fusion rule $f_j^k(\cdot)$ is the minimum Hamming distance fusion rule as given in \eqref{eq:hamm}. Therefore, the performance of the scheme depends on the minimum Hamming distance of the code matrices. Let $d_{min}^k$ be the minimum Hamming distance of the code matrix $C^k$. In the remainder of this section, we analyze our scheme in the asymptotic regime and show that the scheme is asymptotically optimal.

\subsection{Asymptotic Optimality}
\label{sec:a-optimalest}
As before, we study the asymptotic performance of our scheme for two different classes
of tree networks, \textit{fixed height trees} and \textit{fixed degree trees}. We also analyze the scenarios where both the number of nodes and the height of the tree tend to infinity. We first provide the following bound on the misclassification probability at the FC which will be used to prove the asymptotic optimality. Let $Q^k_{m}$ be the probability of misclassifying hypothesis $H_m^k$ at level $k$ and define $q_{max}^k \defeq \max_{1\leq m\leq M} Q_m^k$. For $k=1,\cdots,K$, $Q_m^K=1-Pr(\text{decide $H_m^k$ at level $K | H_m^k$ is true})$.

\begin{proposition}
\label{asy1}
In a perfect tree structure $T(K,N)$ employing the proposed scheme, if $q_{max}^k<\frac{1}{2}$, the misclassification probability at the FC, $P_{e}^{0}$, is bounded as follows
\begin{equation}
\label{asym}
P_e^0\leq 1-\prod\limits_{k=1}^{K}\left[1-(M-1)(4q_{max}^k(1-q_{max}^k))^{\frac{d_{min}^k}{2}}\right].
\end{equation}
\end{proposition}
\begin{proof}
To prove the proposition, we first establish the following set of inequalities (Please see \eqref{17}-\eqref{19})
\begin{eqnarray*}
\sum\limits_{m=1}^{M} P_m^k Pr(\text{decision at level $k-1$}\neq H_m^k\;|\; H_m^k)&\leq&  \sum\limits_{m=1}^{M} P_m^kPr(d^k(\mathbf{u^k},\mathbf{c_m^k})\geq \underset{1 \leq l \leq M,\;l\neq m}{\min}d^k(\mathbf{u^k},\mathbf{c_l^k})\;|\;H_m^k)\nn\\
&\leq&\sum\limits_{m=1}^{M} P_m^k \sum\limits_{\underset{l\neq m}{l=1}}^{M} Pr(d^k(\mathbf{u^k},\mathbf{c_m^k})\geq d^k(\mathbf{u^k},\mathbf{c_l^k})\;|\;H_m^k)\nn\\
&\leq& \sum\limits_{m=1}^{M}P_m^k \sum\limits_{\underset{l\neq m}{l=1}}^{M} \left[\sqrt{4Q_{mm}^{k}(1-Q_{mm}^{k})}\right]^{d^k(\mathbf{c_l^k},\mathbf{c_m^k})}\\
&\leq& (M-1) \left[\sqrt{4q_{max}^{k}(1-q_{max}^{k})}\right]^{d_{min}^k}.
\end{eqnarray*}
Therefore,
\begin{eqnarray*}
\sum\limits_{l=1}^{M} P_m^k Pr(\text{decision at level $k-1$}\neq H_m^k\;|\; H_m^k)&\leq& (M-1) \left[\sqrt{4q_{max}^{k}(1-q_{max}^{k})}\right]^{d_{min}^k}\\
\Leftrightarrow
\sum\limits_{l=1}^{M} P_m^k Pr(\text{decision at level $k-1$}= H_m^k\;|\; H_m^k)&\geq &1-(M-1) \left[\sqrt{4q_{max}^{k}(1-q_{max}^{k})}\right]^{d_{min}^k}.
\end{eqnarray*}

Now, 
\begin{eqnarray*}
P_e^0&=&1-\prod\limits_{k=1}^{K}\sum\limits_{l=1}^{M} P_m^k Pr(\text{decision at level $k-1$}= H_m^k\;|\; H_m^k)\\
&\leq&1-\prod\limits_{k=1}^{K}\left[1-(M-1)(4q_{max}^k(1-q_{max}^k))^{\frac{d_{min}^k}{2}}\right].
\end{eqnarray*}
\end{proof}
As a consequence of Proposition~\eqref{asy1}, for fixed height trees, the probability of `zooming' into the incorrect region of the proposed scheme vanishes as $d_{min}^k$ approaches infinity which happens when $N\to\infty$. 

\begin{small}
\begin{eqnarray*}
\lim\limits_{N\rightarrow\infty} P_e^0&\leq& \lim\limits_{N\rightarrow\infty} \left[ 1-\prod\limits_{k=1}^{K}\left(1-(M-1)(4q_{max}^k(1-q_{max}^k))^{\frac{d_{min}^k}{2}}\right)\right]\\
&=&  \left[ 1-\prod\limits_{k=1}^{K}\lim\limits_{N\rightarrow\infty}\left(1-(M-1)(4q_{max}^k(1-q_{max}^k))^{\frac{d_{min}^k}{2}}\right)\right]\\
&=&   1-\prod\limits_{k=1}^{K}\left[1-(M-1)\lim\limits_{N\rightarrow\infty}\left((4q_{max}^k(1-q_{max}^k))^{\frac{d_{min}^k}{2}}\right)\right]\\
&=&   1-\prod\limits_{k=1}^{K}\left[1-(M-1)0\right]\\
&=& 0.
\end{eqnarray*}
\end{small}

Hence, the overall detection probability becomes `1' as the
degree of the tree $N$ goes to infinity. This shows that
the proposed scheme asymptotically attains perfect region
detection probability for bounded height tree networks if $q_{max}^k<1/2$ $\forall k=1,\cdots,K$. Notice that perfect region detection probability does not imply that the estimation error will vanish. It just provides a coarse estimate of the parameter. For estimation error to vanish, $M^K\rightarrow\infty$, which can be achieved by letting $K$ approach infinity.

However, for fixed degree trees, 
misclassification error of the proposed scheme need not vanish as $K$ approaches infinity.

\begin{eqnarray*}
\lim\limits_{K\rightarrow\infty} P_e^0&\leq& \lim\limits_{K\rightarrow\infty} \left[ 1-\prod\limits_{k=1}^{K}\left(1-(M-1)(4q_{max}^k(1-q_{max}^k))^{\frac{d_{min}^k}{2}}\right)\right]\\
&=&  \left[ 1-\lim\limits_{K\rightarrow\infty}\prod\limits_{k=1}^{K}\left(1-(M-1)(4q_{max}^k(1-q_{max}^k))^{\frac{d_{min}^k}{2}}\right)\right]
\end{eqnarray*}

For misclassification error to vanish, every term in the product should vanish, which obviously is not true for the above equation.

These results can be summarized as the following theorem:

\begin{theorem}
The proposed iterative classification scheme for distributed parameter estimation in tree based networks is asymptotically optimal (when both the degree $N$ and number of levels $K$ simultaneously approach infinity), as long as the probabilities of correct local classification for all hypotheses at each node is greater than one half. 
\end{theorem}

\textit{Remark:} Note that, while for distributed classification, we have shown that the proposed scheme is asymptotically optimal if either $N$ or $K$ tend to infinity, for the distributed estimation case, we have proved that the scheme is asymptotically optimal when both $N$ and $K$ tend to infinity.

Next, we address the remaining aspect of the scheme which is the discretization of the continuous parameter space to perform estimation as iterative classification.

\subsection{Optimal Splitting of the Parameter Space}
\label{sec:split}
As mentioned before, the scheme splits the parameter space $\Theta$ into $M^K$ regions. Therefore, the MSE between the true parameter value $\theta$ and the FC's estimate $\hat{\theta}$ is affected by two factors: the quantization of the continuous region $\Theta$ into $M^K$ discrete points and the probability of misclassifying the region where the true parameter belongs. In Section~\ref{sec:a-optimalest}, we showed that the probability of misclassification can be made to tend to zero by using a large sensor network. Therefore, in order to minimize the MSE, we need to minimize the error due to the quantization of $\Theta$ into $M^K$ points. This optimal splitting depends on the prior pdf $p_\theta(\cdot)$ and can be determined by using ideas from rate distortion theory \cite{cover_inftheory}. As mentioned in \cite{cover_inftheory}, the optimal regions for quantization are given by Voronoi regions and the reconstruction points should minimize the conditional expected distortion over their respective assignments. One of the most popular algorithms used to determine these regions is the Llyod-Max algorithm \cite{lloyd,max}. This algorithm is iterative where we start with an initial set of reconstruction points which are typically chosen at random. It then repeatedly executes the following steps until convergence:
\begin{itemize}
\item Compute the optimal set of reconstruction regions (Voronoi regions) and 
\item Find the set of optimal reconstruction points for these regions (centroid of the Voronoi regions). 
\end{itemize}
In this paper, we use this algorithm which is performed offline and, therefore, is not a computational issue.

\subsection{Simulation Results}
\label{sec:simest}
In this section, we provide simulation results to evaluate the performance of the proposed scheme. As before, consider a tree network $T (3, 7)$ consisting of a total $N_{total} = 400$ nodes, including the FC. The observation at each node is Gaussian distributed with unknown mean $\theta$ and variance $\sigma^2 = 1$. This unknown parameter $\theta$ is uniformly distributed in $(0, \theta_{max} )$ where the region size is varied by varying the maximum value $\theta_{max}$. At each level, the nodes perform an $M$-ary classification where $M = 4$. Therefore, there are a total of $M^K = 4^3 = 64$ possible estimates of $\theta$. Since the parameter is uniformly distributed, the optimal splitting is uniform quantization into $M^K$ regions with the mid-points of the regions as the corresponding representation. Due to the complexity in designing the optimal matrix of size $4\times343$ for transmission at level $3$ (due to collaboration, each node at level $2$ has data of all nodes at level $3$), we employ a sub-optimal approach by concatenating the optimal code matrix of size $4\times7$. For level $3$, it is concatenated $49$ times, and for transmission at level $2$, it is concatenated $7$ times. The smaller code matrix of size $4\times7$ is designed using the simulated annealing approach.

\begin{figure}[t]
\centering
\includegraphics[width=3in,height=!]{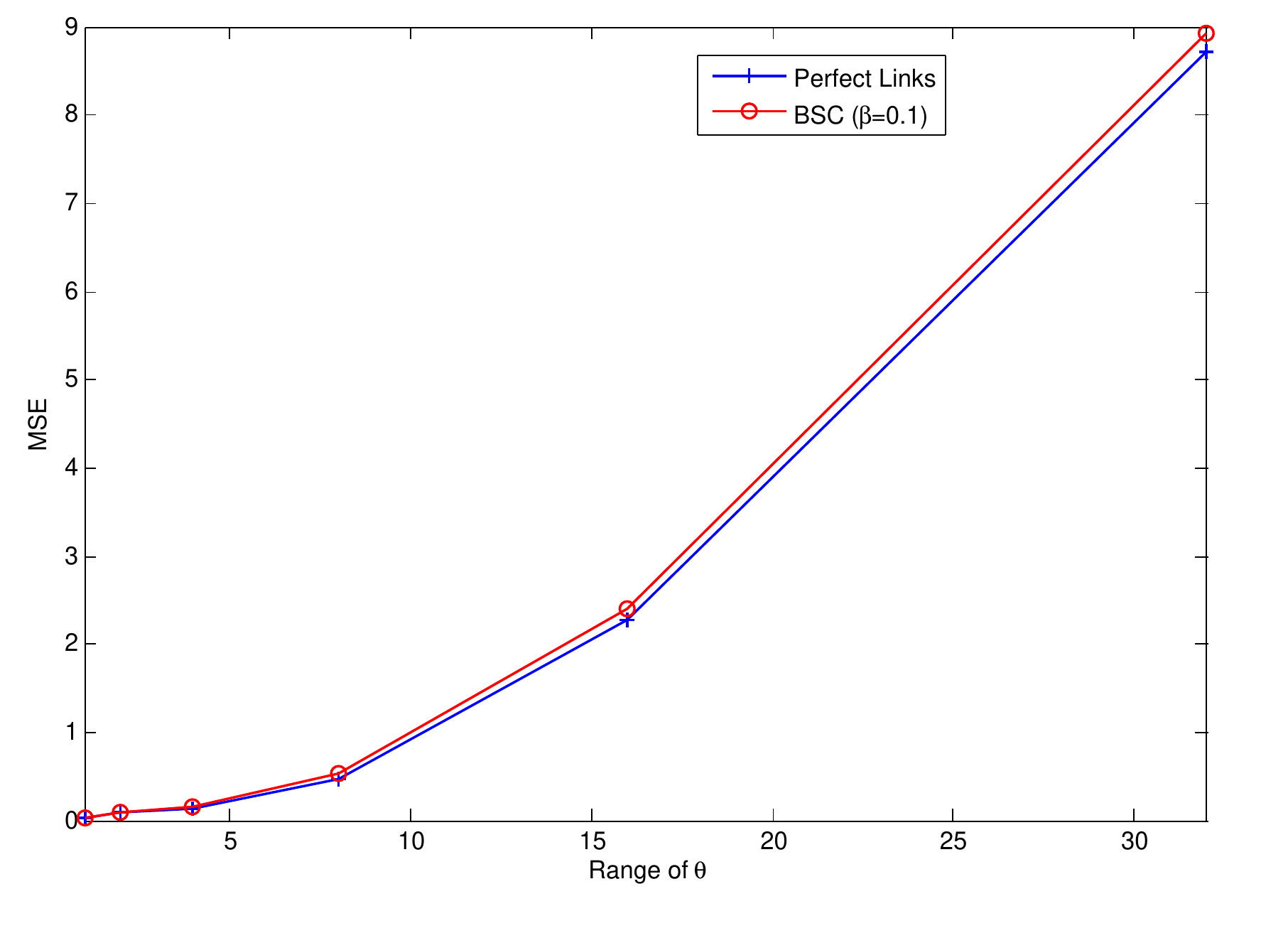}\vspace{-0.25in}
\caption{MSE as a function of the range of $\theta$}
\label{fig:MSE}
\end{figure}

In Figure \ref{fig:MSE}, we plot the mean square error (MSE) between the true value of $\theta$ and its estimate $\hat{\theta}$ at the FC\footnote{As discussed before, this estimate is one of the $M^K$ discrete points representing the quantized regions (centroids of the Voronoi regions, please see Section~\ref{sec:split}).} as a function of $\theta_{max}$. This value of MSE is empirically found by performing $N_{mc} = 5000$ Monte-Carlo runs. As we can observe, the performance of the scheme gets worse with increasing region size. This is because, when the range of $\theta$ is increased while the total number of possible estimates remains fixed, the error due to quantization increases. Since the proposed scheme is based on error-correcting codes, it can tolerate some errors in data. As mentioned before, these errors could be due to various reasons \cite{VempatyHVV14}. We have also simulated the case when the links between the levels are modeled as binary symmetric channels with crossover probability $\beta = 0.1$. As shown in Figure~\ref{fig:MSE}, the proposed scheme is quite robust to the presence of imperfect data arising due to non-ideal channels modeled as binary symmetric channels. As alluded to before, this robustness in performance is due to the use of error-correcting codes.

\section{Conclusion}
\label{sec:disc}
In this paper, we considered the general framework of distributed inference problem in tree networks. We proposed an analytically tractable scheme to solve these problems and proved the asymptotic optimality of the proposed schemes. For the classification problem, when the number of hypotheses is $M=2$, the proposed scheme is a majority-vote scheme for distributed detection in tree networks. Also, note that since the proposed scheme uses error-correcting codes, it works well even in scenarios with unreliable data \cite{VempatyHVV14}. It should be pointed out that the proposed scheme is not limited to wireless sensor networks, although the application of wireless sensor networks has been considered in this paper. The DCFECC scheme has been found to be applicable to a number of other applications including the paradigm of crowdsourcing \cite{VempatyVV2014}. We believe that one can use these results to address several other applications involving tree structures.

\bibliographystyle{IEEEtran}
\bibliography{Conf,Book,Journal}
\end{document}